\documentclass[11pt]{article}
\usepackage[margin=1.05in]{geometry}
\usepackage{amsmath,amssymb,amsthm}
\usepackage[dvipsnames]{xcolor}
\usepackage{graphicx,tikz,multirow,enumerate}
\usepackage[utf8]{inputenc}
\usepackage[T1]{fontenc}
\usepackage[english]{babel}
\usepackage{hyperref}
\hypersetup{colorlinks=true,urlcolor=black,citecolor=blue,linkcolor=blue}

\numberwithin{equation}{section}
\theoremstyle{plain}
\newtheorem{theorem}{Theorem}[section]
\newtheorem{lemma}[theorem]{Lemma}
\newtheorem{proposition}[theorem]{Proposition}
\newtheorem{corollary}[theorem]{Corollary}

\theoremstyle{definition}

\newtheorem{example}[theorem]{Example}
\newtheorem{remark}[theorem]{Remark}
\newtheorem{problem}[theorem]{Problem}

\title{Near-MDS codes of lengths $q+6$ and $q+7$ from conics in $\mathrm{PG}(2,q)$, $q$ odd}

\author{Sinan D\"onmez\thanks{Department of Mathematics, Faculty of Arts and Sciences,
Bursa Uluda\u{g} University, Bursa, T\"urkiye. ORCID: 0009-0000-1319-0630.}
\and Atilla Akp{\i}nar\thanks{Department of Mathematics, Faculty of Arts and Sciences,
Bursa Uluda\u{g} University, Bursa, T\"urkiye. ORCID: 0000-0002-7612-2448.}
\and Rumi Melih Pelen\thanks{Corresponding author. Department of Mathematics and
Statistics, University of South Florida, Tampa, FL, USA. E-mail: rmpelen@usf.edu.
ORCID: 0000-0002-0267-4821.}}
\date{}

\newcommand{\Fq}{\mathbb{F}_q}
\newcommand{\PG}{\mathrm{PG}}
\newcommand{\Code}{\mathcal{C}}
\DeclareMathOperator{\wt}{wt}

\begin{document}

\maketitle
\begin{abstract}
Near maximum distance separable (NMDS) codes of dimension~$3$ and length~$n$ over $\Fq$ are
equivalent to $(n,3)$-arcs in $\PG(2,q)$. For every odd prime power $q$ we construct, by adding
five suitable points to a conic of $\PG(2,q)$, a family of $[q+6,3,q+3]$ NMDS codes and determine
their weight distributions completely; three distinct weight enumerators occur, governed by two
explicit quadratic-character conditions on the parameters. For every odd prime power $q$, no code of our family is
monomially equivalent to a code of the recent $[q+6,3,q+3]$ NMDS family of Fan, Wang and Xu, even where the weight enumerators of the two
families coincide: the separating invariant is a triple of geometric data attached to the
underlying arc. Extending the configuration by a sixth point on a distinguished external line, we
further obtain $[q+7,3,q+4]$ NMDS codes, together with their weight distributions, for every
odd prime power $q\ge11$ (admissible parameters exist for no $q\le9$); the existence proof
combines exact and Weil-type character sum estimates with a finite computer verification. Our proofs
are purely geometric and rest on a simple counting identity for the trisecant lines of a point
set obtained by extending a conic. All the codes constructed are optimal locally recoverable
codes with locality $2$.

\end{abstract}


\section{Introduction}

Let $\Fq$ be the finite field with $q$ elements. A $q$-ary $[n,k]$ linear code $\Code$ is a
$k$-dimensional subspace of $\Fq^n$; its minimum distance $d$ satisfies the Singleton bound
$d\le n-k+1$ \cite{Singleton}. Codes attaining the bound are called maximum distance separable
(MDS). The \emph{Singleton defect} of an $[n,k,d]$ code is $s(\Code)=n-k+1-d$
 \cite{FW97}; codes with $s(\Code)=1$ are called \emph{almost MDS} (AMDS) \cite{deBoer}, and
codes with $s(\Code)=s(\Code^{\perp})=1$, that is, AMDS codes whose dual is also
AMDS, are called \emph{near maximum distance separable} (NMDS). NMDS codes were
introduced by Dodunekov and Landjev \cite{DL95} and are closely related to combinatorial designs \cite{DT20}
and to arcs in finite projective spaces \cite{MMP}; they have found applications in secret sharing \cite{sss} and in
locally recoverable codes \cite{LH23}.

An $[n,3,n-3]$ NMDS code over $\Fq$ is equivalent to an $(n,3)$-arc in $\PG(2,q)$, that is, a set
of $n$ points meeting every line in at most three points and some line in exactly three points
\cite{DL95}. A natural way to produce $(n,3)$-arcs is to extend a maximal arc of $\PG(2,q)$ (an
oval for odd $q$, a hyperoval for even $q$) by adding suitable points. Wang and Heng
\cite{WH21} and Li and Heng \cite{LH23} obtained $[q+3,3,q]$ and $[q+4,3,q+1]$ NMDS codes in this
way, and Zhao, Du and Qiao \cite{ZDQ26} recently obtained further families of lengths up to
$q+3$ together with their locally recoverable properties; Xu, Fan and Han \cite{XFH24} obtained
two families of $[q+5,3,q+2]$ NMDS codes for any prime power $q$ (see also \cite{QDY26} for
families of lengths $q$ to $q+5$ in even characteristic); and, most recently, Fan, Wang and Xu \cite{FWX24} obtained a family of $[q+6,3,q+3]$
NMDS codes for odd $q$ and families of $[q+7,3,q+4]$ and $[q+m+2,3,q+m-1]$ NMDS codes for even
$q=2^m$. In a different direction, Lu and Zhou \cite{LZ25} initiated a systematic study of the
\emph{monomial equivalence} problem for NMDS codes obtained by extending arcs, showing that codes
with identical parameters, and even identical weight enumerators, arising from such constructions may or may not be equivalent.

In this paper we construct, for every odd prime power $q$, a family of $[q+6,3,q+3]$ NMDS codes
by adding five suitable points to a conic of $\PG(2,q)$, and we determine their weight
distributions completely (Theorem~\ref{thm:main}). The family exhibits exactly three weight
enumerators, distinguished by two explicit quadratic-character conditions on the two parameters
$b_1,b_2$ of the construction. Our codes have the same parameters as the odd-$q$ family of
\cite{FWX24}, but they are new: for every odd prime power $q$, no code of our family is monomially
equivalent to a code of theirs (Theorem~\ref{thm:fullineq} and Remark~\ref{rem:smallq}). The separation does not
always come from weight enumerators, which the two families may share (they realize the same
enumerators for $q\equiv1\pmod4$), but from a stronger geometric invariant: for $q\ge9$ any monomial equivalence must map the
conic onto the conic, so the triple $(i,t,e)$ introduced below is preserved; the family of
\cite{FWX24} always has $t=2$, while our construction always has $t\in\{3,4\}$.
In addition, the enumerator $A_{q+3}=(5q+9)(q-1)/2$, attained by our family precisely when
$q\equiv3\pmod4$ (Proposition~\ref{prop:real}), does not occur in \cite{FWX24} at all
(Corollary~\ref{cor:ineq}). Since monomially equivalent codes have equal weight enumerators, the
corresponding codes are inequivalent to those of \cite{FWX24}. In the terminology of \cite{LZ25},
our construction thus enlarges the known set of equivalence classes of $[q+6,3,q+3]$ NMDS codes.

Two features of our approach may be of independent interest. First, all proofs are purely
geometric, in contrast with the algebraic computations of \cite{XFH24,FWX24}. Second, the weight
distribution is obtained in a single step from a simple counting identity
(Lemma~\ref{lem:identity}): if $S$ is a $(q+6,3)$-arc consisting of a conic $\mathcal{O}$ together
with a $5$-set $T$ of further points, then the number $\tau_3$ of trisecant lines of $S$
(lines meeting $S$ in exactly three points) equals
\[
\tau_3=\frac{5(q-1)}{2}+i+t+e,
\]
where $i$ is the number of points of $T$ interior to $\mathcal{O}$, $t$ is the number of tangent
lines of $\mathcal{O}$ containing two points of $T$, and $e$ is the number of external lines
containing three points of $T$. Since $A_{q+3}=\tau_3(q-1)$ and $A_{q+3}$ determines the whole
weight distribution of an NMDS code \cite{DL95}, the identity reduces the computation of weight
enumerators to reading off the triple $(i,t,e)$ from the configuration.

Our second contribution concerns length $q+7$. For even $q$, the family of \cite{FWX24}
reaches length $q+7$, but for odd $q$ no NMDS codes of length $q+7$ arising from arcs that
extend a maximal arc appear to be known. (NMDS codes of length up to roughly $q+2\sqrt q$ do
exist via elliptic curves \cite{AGS21}; here we are concerned with explicit arc-extension
constructions and their weight distributions.) We show (Theorem~\ref{thm:q7}) that whenever the line
$l_{A,B}$ of our configuration is external, any point $P_{6}$ on it satisfying three further
quadratic-character conditions yields a $[q+7,3,q+4]$ NMDS code whose weight distribution is
again read off from the counting identity.
We prove (Proposition~\ref{prop:q7exist}) that admissible parameters exist for every odd
prime power $q\ge11$, and for none with $q\le9$.

The paper is organized as follows. Section~\ref{sec:pre} collects preliminaries and proves the
counting identity. Section~\ref{sec:main} presents the construction and the main theorem.
Section~\ref{sec:comp} compares our family with that of \cite{FWX24} and gives Magma-verified
examples. Section~\ref{sec:ext} presents the extension to length $q+7$.
Section~\ref{sec:conc} concludes with open problems. All computer verifications reported in
this paper were carried out in Magma \cite{Magma} and Python (see the data availability
statement).

\section{Preliminaries}\label{sec:pre}

Throughout, $q=p^m$ is an odd prime power and $\eta$ denotes the quadratic character of
$\Fq$: $\eta(x)=1$ if $x$ is a nonzero square, $\eta(x)=-1$ if $x$ is a nonsquare, and
$\eta(0)=0$.

\subsection{NMDS codes and $(n,3)$-arcs}

The \emph{weight} $\wt(x)$ of $x\in\Fq^{n}$ is the number of its nonzero coordinates; the
minimum distance of a linear code equals the minimum nonzero weight of its codewords. The
\emph{dual} of a code $\Code\subseteq\Fq^{n}$ is
$\Code^{\perp}=\{x\in\Fq^{n}: x\cdot c=0\ \text{for all}\ c\in\Code\}$. For $0\le i\le n$ we write
$A_{i}$ for the number of codewords of $\Code$ of weight $i$; the polynomial
$A(y)=1+\sum_{i=1}^{n}A_{i}y^{i}$ is the \emph{weight enumerator} of $\Code$. Two codes are
\emph{monomially equivalent} (also called \emph{linearly equivalent}, as in \cite{FWX24}) if
one is obtained from the other by permuting the coordinates
and multiplying each coordinate by a nonzero scalar; equivalently, if
$\Code_{2}=\{cM: c\in\Code_{1}\}$ for an $n\times n$ \emph{monomial matrix} $M$, that is, a
matrix with exactly one nonzero entry in each row and each column. Monomially equivalent codes
have the same weight enumerator.

Let $\Code$ be an $[n,k,n-k]$ NMDS code with generator matrix $G=(g_1,\dots,g_n)$, $g_i\in\Fq^k$,
$k\ge 3$. The columns of $G$ may be identified with distinct points of $\PG(k-1,q)$; write
$S_G=\{g_1,\dots,g_n\}$. For a nonzero $z\in\Fq^k$ let
$H_z=\{x\in\Fq^k: x\cdot z=0\}$ denote the corresponding hyperplane of $\PG(k-1,q)$.

\begin{lemma}\label{lem:weightline}
For nonzero $z\in\Fq^k$ we have $\wt(zG)=n-\#(H_z\cap S_G)$. In particular every hyperplane meets
$S_G$ in at most $n-d$ points.
\end{lemma}

\begin{proof}
The $i$-th coordinate of $zG$ is $g_{i}\cdot z$, so
$\wt(zG)=n-\#\{i: g_{i}\in H_{z}\}$. Since $d(\Code^{\perp})=k\ge3$, no column of $G$ is zero and
no two columns are proportional (either would produce a dual codeword of weight $\le2$), so
the $g_{i}$ represent $n$ distinct points of $\PG(k-1,q)$ and
$\#\{i: g_{i}\in H_{z}\}=\#(H_{z}\cap S_{G})$. The last assertion follows from
$\wt(zG)\ge d$.
\end{proof}

Recall that an $(n,3)$-\emph{arc} of $\PG(2,q)$ is a set of $n$ points meeting every line in
at most three points and some line in exactly three points \cite{DL95}. Thus an $[n,3,n-3]$
NMDS code corresponds to an $(n,3)$-arc $S_G$ of $\PG(2,q)$: every line meets $S_G$ in at most
$n-d=3$ points by Lemma~\ref{lem:weightline}, and a line meeting $S_G$ in exactly three points
exists because $\Code$ attains the weight $n-3$. Moreover
\begin{equation}\label{eq:Amin}
A_{n-3}=(q-1)\cdot\#\{\ell\ \text{a line of } \PG(2,q):\ \#(\ell\cap S_G)=3\},
\end{equation}
since each trisecant line $\ell=H_z$ accounts for the $q-1$ codewords $\lambda zG$,
$\lambda\in\Fq^{*}$, of minimum weight. The full weight distribution follows from $A_{n-3}$:

\begin{lemma}[{\cite{DL95}, Theorem 4.1}]\label{lem:DL}
Let $\Code$ be an $[n,k,n-k]$ NMDS code. Then, for $s\in\{1,\dots,k\}$,
\[
A_{n-k+s}=\binom{n}{k-s}\sum_{j=0}^{s-1}(-1)^{j}\binom{n-k+s}{j}\bigl(q^{s-j}-1\bigr)
+(-1)^{s}\binom{k}{s}A_{n-k}.
\]
\end{lemma}

For $k=3$, which is the only case used below, Lemma~\ref{lem:DL} reads
\begin{equation}\label{eq:DL3}
\begin{aligned}
A_{n-2}&=\tbinom{n}{2}(q-1)-3A_{n-3},\\
A_{n-1}&=n(q-1)(q+2-n)+3A_{n-3},\\
A_{n}&=(q^{3}-1)-n(q^{2}-1)+\tbinom{n}{2}(q-1)-A_{n-3};
\end{aligned}
\end{equation}
all weight enumerators in this paper are obtained by substituting the relevant value of
$A_{n-3}$ into \eqref{eq:DL3}.

\subsection{Conics and their point--line incidences}

For odd $q$ every oval of $\PG(2,q)$ is a conic (Segre's theorem \cite{Hirschfeld}); we work with
\begin{equation}\label{eq:conic}
\mathcal{O}=\bigl\{(x^{2},x,1): x\in\Fq\bigr\}\cup\{(1,0,0)\},
\end{equation}
a conic of $q+1$ points. Each line of $\PG(2,q)$ meets $\mathcal{O}$ in $0$, $1$ or $2$ points and
is called \emph{external}, \emph{tangent} or \emph{secant} accordingly. A point
$P\notin\mathcal{O}$ is \emph{exterior} if it lies on two tangents and \emph{interior} if it lies
on none; through an exterior point there pass $\tfrac{q-1}{2}$ secants, $\tfrac{q-1}{2}$ external
lines and $2$ tangents, while through an interior point there pass $\tfrac{q+1}{2}$ secants and
$\tfrac{q+1}{2}$ external lines \cite{Hirschfeld}.

We use the following affine model: identify $(x,y)$ with the projective point $(x,y,1)$,
write $(a)=(a,1,0)$ and $(\infty)=(1,0,0)$ for the points of the line at infinity
$L_{\infty}$, and for $a,b\in\Fq$ set
\[
l_{a,b}=\{(av+b,\,v): v\in\Fq\}\cup\{(a)\},\qquad
L_{a}=\{(v,a): v\in\Fq\}\cup\{(\infty)\}.
\]
Then $\mathcal{O}=\{(v^{2},v):v\in\Fq\}\cup\{(\infty)\}$, the line $L_{\infty}$ is the tangent at
$(\infty)$, and $l_{a,b}$ meets the affine part of $\mathcal{O}$ in the points $(v^2,v)$ with
$v^{2}-av-b=0$. Hence $l_{a,b}$ is secant, tangent or external according as
$\eta(a^{2}+4b)=1$, $a^{2}+4b=0$ or $\eta(a^{2}+4b)=-1$. Similarly, the tangents through an
affine point $(u_{0},v_{0})\notin\mathcal{O}$ are the lines $l_{a,u_0-av_0}$ with
$a^{2}-4v_{0}a+4u_{0}=0$, whose discriminant is $16(v_{0}^{2}-u_{0})$; therefore
\begin{equation}\label{eq:inout}
(u_{0},v_{0})\ \text{is exterior}\iff \eta\bigl(v_{0}^{2}-u_{0}\bigr)=1,\qquad
\text{interior}\iff \eta\bigl(v_{0}^{2}-u_{0}\bigr)=-1 .
\end{equation}

We shall repeatedly use the following standard facts on quadratic character sums: a Weil-type
bound and two classical exact evaluations; see
\cite[Theorems 5.41 and 5.48]{LN97}.

\begin{lemma}\label{lem:weil}
Let $f\in\Fq[x]$ have $m$ distinct roots in an algebraic closure $\overline{\Fq}$ of $\Fq$, and suppose $f$ is not a constant
multiple of a square in $\Fq[x]$. Then
\[
\Bigl|\sum_{x\in\Fq}\eta\bigl(f(x)\bigr)\Bigr|\le(m-1)\sqrt q .
\]
Moreover, the following exact evaluations hold: $\sum_{x\in\Fq}\eta(ax+b)=0$ for $a\neq0$, and
$\sum_{x\in\Fq}\eta(ax^{2}+bx+c)=-\eta(a)$ whenever $a\neq0$ and $b^{2}-4ac\neq0$.
\end{lemma}

\subsection{A counting identity for trisecants}

\begin{lemma}\label{lem:identity}
Let $\mathcal{O}$ be a conic of $\PG(2,q)$, $q$ odd, let $T$ be a set of $m\ge 1$ points of
$\PG(2,q)\setminus\mathcal{O}$, and suppose $S=\mathcal{O}\cup T$ is a $(q+1+m,3)$-arc. Let
\begin{gather*}
i=\#\{P\in T: P\ \text{interior to}\ \mathcal{O}\},\qquad
t=\#\{\ell\ \text{tangent}: \#(\ell\cap T)=2\},\\
e=\#\{\ell\ \text{external}: \#(\ell\cap T)=3\}.
\end{gather*}
Then the number of trisecant lines of $S$ is
\[
\tau_{3}:=\#\{\ell: \#(\ell\cap S)=3\}=\frac{m(q-1)}{2}+i+t+e .
\]
\end{lemma}

\begin{proof}
No three points of $\mathcal{O}$ are collinear, so every trisecant of $S$ contains $1$, $2$ or
$3$ points of $T$. A trisecant containing exactly one point of $T$ contains two points of
$\mathcal{O}$, hence is a secant of $\mathcal{O}$; conversely, every secant of $\mathcal{O}$
through a point of $T$ carries exactly three points of $S$, because a fourth point of $S$ on it
would violate the arc condition. Moreover no secant contains two points of $T$ (again four
points of $S$ would be collinear), so these trisecants are counted without repetition by
$\sum_{P\in T}\sigma(P)$, where $\sigma(P)$ is the number of secants of $\mathcal{O}$ through
$P$; by the incidence counts recalled above, $\sigma(P)=\tfrac{q-1}{2}$ or $\tfrac{q+1}{2}$
according as $P$ is exterior or interior, whence
$\sum_{P\in T}\sigma(P)=\tfrac{m(q-1)}{2}+i$.

A trisecant containing exactly two points of $T$ contains exactly one point of $\mathcal{O}$,
i.e.\ is a tangent; conversely a tangent carrying two points of $T$ is a trisecant (it cannot
carry a third point of $T$ by the arc condition). These are counted by $t$. Finally, a trisecant
containing three points of $T$ contains no point of $\mathcal{O}$, i.e.\ is external, and
conversely; these are counted by $e$. Summing the three contributions gives the identity.
\end{proof}

\section{The construction and the main theorem}\label{sec:main}

Fix $b_{1},b_{2}\in\Fq$ and consider the following five points:
\[
P_{1}=(1),\quad P_{2}=(0),\quad
P_{3}=\Bigl(b_{1},\,b_{1}+\tfrac14\Bigr),\quad
P_{4}=\bigl(b_{1},\,b_{1}-b_{2}\bigr),\quad
P_{5}=\bigl(0,\,-b_{2}\bigr).
\]
Geometrically: $P_{1},P_{2}$ lie on the tangent $L_{\infty}$. The lines through $P_{1}=(1)$
other than $L_{\infty}$ are the $l_{1,b}$, $b\in\Fq$, and by the criterion of
Section~\ref{sec:pre} the line $l_{1,b}$ is tangent precisely when $1+4b=0$; thus the second
tangent through $P_{1}$ is $l_{1,-1/4}$, touching $\mathcal{O}$ at $(1/4,1/2)$, the double root
of $v^{2}-v+\tfrac14$. Similarly the lines through $P_{2}=(0)$ other than $L_{\infty}$ are the
$l_{0,b}$, tangent precisely when $b=0$, so the second tangent through $P_{2}$ is $l_{0,0}$,
touching at the origin. Further, $P_{3}=l_{1,-1/4}\cap l_{0,b_{1}}$,
$P_{4}=l_{0,b_{1}}\cap l_{1,b_{2}}$ and $P_{5}=l_{0,0}\cap l_{1,b_{2}}$, where $l_{0,b_{1}}$ and
$l_{1,b_{2}}$ will be external lines. The configuration is depicted in
Figure~\ref{fig:config}.

\begin{figure}[htb]
\centering
\begin{tikzpicture}[scale=0.88,
  tang/.style={blue!70!black, thick},
  seca/.style={green!55!black, thick},
  extl/.style={red!75!black, thick},
  pt/.style={circle, fill=black, inner sep=1.4pt},
  cpt/.style={circle, fill=black, inner sep=1.1pt},
  lb/.style={font=\small}]
\draw[thick] (0,0) circle (1);
\draw[tang] (-2.6,1) -- (2.6,1) node[lb, right] {$L_{\infty}$};
\draw[tang] (-2.0,1.464) -- (0.88,-3.524);            
\node[lb, tang, anchor=west] at (0.95,-3.55) {$l_{1,-1/4}$};
\draw[tang] (2.0,1.464) -- (-0.88,-3.524);            
\node[lb, tang, anchor=east] at (-0.95,-3.55) {$l_{0,0}$};
\draw[extl] (2.0,1.804) -- (-0.11,-4.526);            
\node[lb, extl] at (2.0,2.08) {$l_{0,b_{1}}$};
\draw[extl] (-2.0,1.804) -- (0.11,-4.526);            
\node[lb, extl] at (-2.0,2.08) {$l_{1,b_{2}}$};
\draw[extl, dashed] (-1.35,-2.804) -- (1.35,-2.804);  
\node[lb, extl, anchor=west] at (1.5,-2.804) {$l_{A,B}$};
\draw[seca] (-1.95,1.0765) -- (1.75,-0.2185);
\node[lb, seca, anchor=west] at (1.8,-0.22) {secant};
\node[pt] at (-1.732,1) {}; \node[lb, anchor=east] at (-1.95,0.78) {$P_{1}$};
\node[pt] at (1.732,1) {};  \node[lb, anchor=west] at (1.95,0.78) {$P_{2}$};
\node[pt] at (0.464,-2.804) {};  \node[lb, anchor=west] at (0.68,-2.58) {$P_{3}$};
\node[pt] at (0,-4.196) {};      \node[lb, anchor=north] at (0,-4.65) {$P_{4}$};
\node[pt] at (-0.464,-2.804) {}; \node[lb, anchor=east] at (-0.68,-2.58) {$P_{5}$};
\node[cpt] at (0,1) {};          \node[lb, anchor=north] at (0,0.86) {$O_{q+1}$};
\node[cpt] at (-0.866,-0.5) {};  \node[lb, anchor=west] at (-0.72,-0.40) {$O_{j}$};
\node[cpt] at (0.866,-0.5) {};   \node[lb, anchor=east] at (0.72,-0.40) {$O_{1}$};
\begin{scope}[shift={(-3.7,-3.2)}]
\draw[tang] (0,0) -- (0.55,0); \node[lb, anchor=west] at (0.6,0) {tangent};
\draw[seca] (0,-0.42) -- (0.55,-0.42); \node[lb, anchor=west] at (0.6,-0.42) {secant};
\draw[extl] (0,-0.84) -- (0.55,-0.84); \node[lb, anchor=west] at (0.6,-0.84) {external};
\end{scope}
\end{tikzpicture}
\caption{The configuration of Theorem~\ref{thm:main}, drawn in the real plane after a projective
change of coordinates (so that the tangent $L_{\infty}$ becomes an ordinary line). Blue lines are
tangent to the conic, red lines are external (passants), and one of the $\frac{q-1}{2}$ secants
through $P_{1}$ is shown in green; every secant through a point of $T$ is a trisecant of $S$. The
tangents $L_{\infty}$, $l_{1,-1/4}$, $l_{0,0}$ each carry two points of $T$, and the external
lines $l_{0,b_{1}}$, $l_{1,b_{2}}$ each carry three. The line $l_{A,B}$ joining $P_{3}$ and
$P_{5}$ (dashed) is external when $\eta(\Delta)=-1$, as drawn, and tangent when $\Delta=0$;
likewise $P_{4}$ is drawn exterior, which corresponds to $\eta(\delta)=1$. The tangency points
$(\infty)$, $(1/4,1/2)$ and $(0,0)$ are labelled $O_{q+1}$, $O_{j}$ and $O_{1}$ respectively.}
\label{fig:config}
\end{figure}
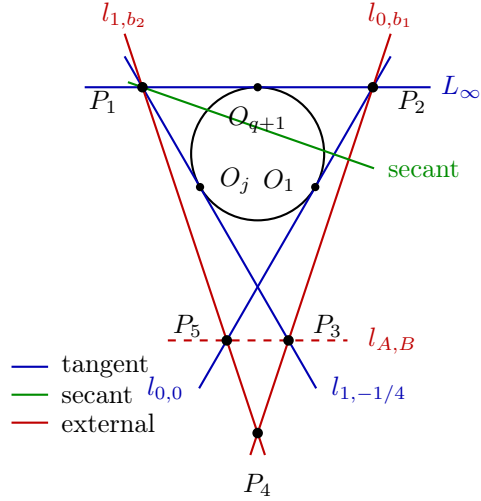

Set
\begin{equation}\label{eq:Delta}
\Delta:=b_{1}^{2}+4\bigl(b_{1}^{2}b_{2}+b_{1}b_{2}^{2}\bigr)+b_{1}b_{2},
\qquad
\delta:=(b_{1}-b_{2})^{2}-b_{1}.
\end{equation}
Here $\Delta$ is, up to the square factor $\bigl(b_1+b_2+\tfrac14\bigr)^{2}$, the discriminant of
the line $l_{A,B}$ joining $P_{3}$ and $P_{5}$, and $\delta$ is, up to the square factor $16$,
the discriminant governing the tangents through $P_{4}$. A direct computation gives the
factorization
\begin{equation}\label{eq:Deltafact}
\Delta=4\,b_{1}\Bigl(b_{2}+\tfrac14\Bigr)(b_{1}+b_{2}),
\end{equation}
so that under the hypotheses $\eta(b_{1})=-1$, $\eta(1+4b_{2})=-1$ one has $\Delta=0$ if and
only if $b_{2}=-b_{1}$, and $\eta(\Delta)=-\eta\bigl((b_{2}+\tfrac14)(b_{1}+b_{2})\bigr)$.
Moreover, by \eqref{eq:inout},
\begin{equation}\label{eq:P4}
P_{4}\ \text{is exterior}\iff \eta(\delta)=1,\qquad
P_{4}\ \text{is interior}\iff \eta(\delta)=-1 .
\end{equation}
Note that $\eta(b_{1})=-1$ forces $\delta\neq0$, since $\delta=0$ would make
$b_{1}=(b_{1}-b_{2})^{2}$ a square.

Define the $3\times(q+6)$ matrix
\begin{equation}\label{eq:G}
G_{(b_{1},b_{2})}=
\begin{pmatrix}
\alpha_{1}^{2} & \cdots & \alpha_{q}^{2} & 1 & 1 & 0 & b_{1} & b_{1} & 0\\[2pt]
\alpha_{1} & \cdots & \alpha_{q} & 0 & 1 & 1 & \dfrac{4b_{1}+1}{4} & b_{1}-b_{2} & -b_{2}\\[4pt]
1 & \cdots & 1 & 0 & 0 & 0 & 1 & 1 & 1
\end{pmatrix},
\end{equation}
where $\Fq=\{\alpha_{1},\dots,\alpha_{q}\}$, and let $\Code_{(b_{1},b_{2})}$ be the linear code over
$\Fq$ generated by $G_{(b_{1},b_{2})}$. The first $q+1$ columns are the conic \eqref{eq:conic}
and the last five columns are $P_{1},\dots,P_{5}$.

\begin{theorem}\label{thm:main}
Let $q$ be an odd prime power and let $b_{1},b_{2}\in\Fq$ satisfy
\[
\eta(b_{1})=-1,\qquad \eta(1+4b_{2})=-1,\qquad b_{1}+b_{2}\neq-\tfrac14,\qquad
\Delta=0\ \text{ or }\ \eta(\Delta)=-1 .
\]
Then $\Code_{(b_{1},b_{2})}$ is a $[q+6,3,q+3]$ NMDS code over $\Fq$ with
\[
A_{q+3}=\Bigl(\tfrac{5(q-1)}{2}+i+t+e\Bigr)(q-1),
\qquad
(i,t,e)=\Bigl(\tfrac{1-\eta(\delta)}{2},\ 3+[\Delta=0],\ 2\Bigr)
\]
(here $[\,\cdot\,]$ denotes the Iverson bracket, equal to $1$ if the enclosed condition holds
and to $0$ otherwise); that is,
\[
A_{q+3}=
\begin{cases}
\dfrac{(5q+9)(q-1)}{2}, & \Delta=0,\ \eta(\delta)=-1,\\[6pt]
\dfrac{(5q+7)(q-1)}{2}, & \Delta=0,\ \eta(\delta)=1,\quad\text{or}\quad
\eta(\Delta)=-1,\ \eta(\delta)=-1,\\[6pt]
\dfrac{(5q+5)(q-1)}{2}, & \eta(\Delta)=-1,\ \eta(\delta)=1 .
\end{cases}
\]
The corresponding weight enumerators are, respectively,
\begin{align*}
A(y)&=1+\tfrac{(5q+9)(q-1)}{2}\,y^{q+3}+\tfrac{(q-1)(q^{2}-4q+3)}{2}\,y^{q+4}
+\tfrac{(q-1)(7q-21)}{2}\,y^{q+5}+\tfrac{(q-1)(q^{2}-6q+11)}{2}\,y^{q+6},\\[3pt]
A(y)&=1+\tfrac{(5q+7)(q-1)}{2}\,y^{q+3}+\tfrac{(q-1)(q^{2}-4q+9)}{2}\,y^{q+4}
+\tfrac{(q-1)(7q-27)}{2}\,y^{q+5}+\tfrac{(q-1)(q^{2}-6q+13)}{2}\,y^{q+6},\\[3pt]
A(y)&=1+\tfrac{(5q+5)(q-1)}{2}\,y^{q+3}+\tfrac{(q-1)(q^{2}-4q+15)}{2}\,y^{q+4}
+\tfrac{(q-1)(7q-33)}{2}\,y^{q+5}+\tfrac{(q-1)(q^{2}-6q+15)}{2}\,y^{q+6}.
\end{align*}
\end{theorem}

\begin{proof}
Write $T=\{P_{1},\dots,P_{5}\}$ and $S=\mathcal{O}\cup T$.

\emph{Step 1: the points are well defined, pairwise distinct and off the conic.}
The points $P_{1},P_{2}$ are distinct points of $L_{\infty}$, hence differ from the affine
points $P_{3},P_{4},P_{5}$, and they are off $\mathcal{O}$ since the only point of
$\mathcal{O}$ on $L_{\infty}$ is $(\infty)$. Since $\eta(b_{1})=-1$ we have $b_{1}\neq0$, so
$P_{3}$ and $P_{4}$, whose first coordinate is $b_{1}$, both differ from $P_{5}$, whose first
coordinate is $0$; and $P_{3}\neq P_{4}$ because $b_{2}\neq-\tfrac14$
(indeed $\eta(1+4b_{2})=-1$). The point $P_{3}$ is off $\mathcal{O}$ since
$b_{1}=(b_{1}+\tfrac14)^{2}$ would make $b_{1}$ a square; $P_{4}$ is off $\mathcal{O}$ since
$\delta\neq0$; and $P_{5}$ is off $\mathcal{O}$ since $b_{2}=0$ is excluded by
$\eta(1+4b_{2})=-1$.

\emph{Step 2: every pair of points of $T$ lies on a tangent or external line.}
The six distinguished lines and the points of $T$ they carry are:
\[
\begin{array}{lll}
L_{\infty} & \text{tangent at } (\infty) & \ni P_{1},P_{2};\\
l_{1,-1/4} & \text{tangent at } (1/4,1/2) & \ni P_{1},P_{3};\\
l_{0,0} & \text{tangent at } (0,0) & \ni P_{2},P_{5};\\
l_{0,b_{1}} & \text{external, since } \eta(4b_{1})=\eta(b_{1})=-1 & \ni P_{2},P_{3},P_{4};\\
l_{1,b_{2}} & \text{external, since } \eta(1+4b_{2})=-1 & \ni P_{1},P_{4},P_{5};\\
l_{A,B} & \text{tangent or external, by } \Delta=0 \text{ or } \eta(\Delta)=-1 & \ni P_{3},P_{5},
\end{array}
\]
where $l_{A,B}$ is the line joining $P_{3}$ and $P_{5}$: their second coordinates
$b_{1}+\tfrac14$ and $-b_{2}$ differ precisely because $b_{1}+b_{2}\neq-\tfrac14$, so the join
is not a line $L_{a}$ through $(\infty)$, and solving
$A\bigl(b_{1}+\tfrac14\bigr)+B=b_{1}$ and $-Ab_{2}+B=0$ gives
$A=\dfrac{b_{1}}{\,b_{1}+b_{2}+\tfrac14\,}$, $B=Ab_{2}$, and
\[
A^{2}+4B=\frac{b_{1}^{2}+4b_{1}b_{2}\bigl(b_{1}+b_{2}+\tfrac14\bigr)}
{\bigl(b_{1}+b_{2}+\tfrac14\bigr)^{2}}
=\frac{\Delta}{\bigl(b_{1}+b_{2}+\tfrac14\bigr)^{2}} .
\]
(Here $P_{2}\in l_{0,b_{1}}$ and $P_{1}\in l_{1,b_{2}}$ because $(0)$ and $(1)$ are the points at
infinity of those lines.) These six lines cover all ten pairs of $T$. Consequently no secant of
$\mathcal{O}$ contains two points of $T$.

\emph{Step 3: no four points of $S$ are collinear.}
By Step 2 the unique line through each pair of points of $T$ is one of the six lines listed
there; every other line therefore carries at most one point of $T$, hence at most three points
of $S$. A line carrying four points of $S$ would thus be one of the six lines; since its
points of $\mathcal{O}$ and its listed points of $T$ total at most three ($1+2$ for a tangent,
$0+3$ for an external line), it would contain a point $X\in T$ not listed on it. Pairing $X$
with a listed point $Y$ of that
line, the pair $\{X,Y\}$ lies on the line assigned to it in Step 2, a different entry of the
list, and since only one line passes through $X$ and $Y$, two of the six lines would coincide.
It thus suffices to show that the six lines are pairwise distinct. For the five explicit lines
this is immediate: $L_{\infty}$ is the line at infinity, the lines $l_{1,\cdot}$ and
$l_{0,\cdot}$ have the distinct points at infinity $(1)$ and $(0)$, and
$l_{1,-1/4}\neq l_{1,b_{2}}$, $l_{0,0}\neq l_{0,b_{1}}$ since a tangent is never external.
Finally, $P_{3}$ lies on $l_{1,-1/4}$ and $l_{0,b_{1}}$, and $P_{5}$ on $l_{1,b_{2}}$ and
$l_{0,0}$; if $l_{A,B}$ coincided with one of the five explicit lines, then $P_{3}$ or $P_{5}$
would lie on two distinct lines of the pencil through $P_{1}$ or through $P_{2}$, hence equal
$P_{1}$ or $P_{2}$, which is absurd for an affine point. Hence $S$ is a $(q+6,3)$-arc (three collinear points exist, e.g.\ on
$l_{0,b_{1}}$), and by Lemma~\ref{lem:weightline} the code $\Code_{(b_{1},b_{2})}$ is a
$[q+6,3,q+3]$ NMDS code.

\emph{Step 4: the triple $(i,t,e)$.}
The points $P_{1},P_{2},P_{3},P_{5}$ each lie on a tangent of $\mathcal{O}$ (Step 2), hence are
exterior; $P_{4}$ is exterior or interior according to \eqref{eq:P4}, so
$i=\bigl(1-\eta(\delta)\bigr)/2$. The tangents carrying two points of $T$ are exactly
$L_{\infty}$, $l_{1,-1/4}$, $l_{0,0}$, together with $l_{A,B}$ when $\Delta=0$; no other tangent
can carry two points of $T$ since all pairs are accounted for in Step 2. Hence
$t=3+[\Delta=0]$. The external lines carrying three points of $T$ are exactly $l_{0,b_{1}}$ and
$l_{1,b_{2}}$ (a third such line would carry a pair not listed in Step 2 or coincide with
$l_{A,B}$, which carries only $P_{3},P_{5}$), so $e=2$. Lemma~\ref{lem:identity} and
\eqref{eq:Amin} now give the stated values of $A_{q+3}$, and Lemma~\ref{lem:DL} yields the three
weight enumerators.
\end{proof}

\begin{proposition}\label{prop:real}
Suppose $\eta(b_{1})=-1$ and $\eta(1+4b_{2})=-1$. If $\Delta=0$, then $b_{2}=-b_{1}$,
$\delta=b_{1}(4b_{1}-1)$ and $\eta(\delta)=\eta(-1)$. Consequently the case
$\Delta=0$, $\eta(\delta)=-1$ of Theorem~\ref{thm:main} is realizable precisely when
$q\equiv3\pmod4$, and the case $\Delta=0$, $\eta(\delta)=1$ precisely when
$q\equiv1\pmod4$; in either case the number of admissible parameters $b_{1}$ equals
$\bigl(q-\eta(-1)\bigr)/4$.
\end{proposition}

\begin{proof}
By \eqref{eq:Deltafact}, $\Delta=0$ forces $b_{2}=-b_{1}$, whence
$\delta=(2b_{1})^{2}-b_{1}=b_{1}(4b_{1}-1)$ and
$\eta(\delta)=\eta(b_{1})\,\eta(4b_{1}-1)
=-\,\eta(-1)\,\eta(1-4b_{1})=\eta(-1)$,
using $\eta(1+4b_{2})=\eta(1-4b_{1})=-1$. Thus $\eta(\delta)$ is forced to equal $\eta(-1)$;
since $\eta(-1)=1$ for $q\equiv1\pmod4$ and $\eta(-1)=-1$ for $q\equiv3\pmod4$, the first two
assertions follow. For the count, the admissible $b_{1}$ are those with
$\eta(b_{1})=-1$ and $\eta(1-4b_{1})=-1$ (the condition $b_{1}+b_{2}=0\neq-\tfrac14$ is
automatic), and their number is
\[
\frac14\sum_{b_{1}\in\Fq}\bigl(1-\eta(b_{1})\bigr)\bigl(1-\eta(1-4b_{1})\bigr)
=\frac14\Bigl(q+\sum_{b_{1}}\eta\bigl(b_{1}(1-4b_{1})\bigr)\Bigr)
=\frac{q-\eta(-1)}{4},
\]
by Lemma~\ref{lem:weil} applied with $a=-4$; the terms $b_{1}\in\{0,\tfrac14\}$ contribute
$0$ to the sum on the left.
\end{proof}

\begin{remark}\label{rem:vacuous}
Proposition~\ref{prop:real} explains the vacuous cases observed computationally: the case
$\Delta=0$, $\eta(\delta)=-1$ admits no parameters when $q\equiv1\pmod4$ (e.g.\ $q=5,9$) and
the case $\Delta=0$, $\eta(\delta)=1$ none when $q\equiv3\pmod4$ (e.g.\ $q=7$). An exhaustive
computer search over all $(q+6,3)$-arcs containing a conic shows the stronger fact that for
$q=5$ the value $A_{q+3}=(5q+9)(q-1)/2$ is not attained by any such arc, whether or not it
arises from our construction. The exact number of
admissible pairs $(b_{1},b_{2})$ in the remaining cases $\eta(\Delta)=-1$ can be determined by
similar character sum computations; we do not pursue this here.
\end{remark}

\section{Comparison with the Fan--Wang--Xu family}\label{sec:comp}

Fan, Wang and Xu \cite{FWX24} constructed a family of $[q+6,3,q+3]$ NMDS codes for every odd
prime power $q$ and determined its weight distributions. In our affine notation their
construction reads as follows: for $\beta,\gamma\in\Fq^{*}$ with
\[
\eta(\beta)=\eta(1+4\beta)=\eta(1+4\beta-4\gamma)=\eta(1+4\beta+4\gamma)=-1,
\]
the arc is $S_{F}=\mathcal{O}\cup T_{F}$ with
$T_{F}=\bigl\{(1),\,(0),\,(0,-\beta),\,(\beta,0),\,(\beta,\gamma)\bigr\}$; the
configuration is depicted in Figure~\ref{fig:fwx}, in the same projective frame as
Figure~\ref{fig:config}. Lemma~\ref{lem:identity} recovers their weight distributions at once:

\begin{proposition}\label{prop:fwxprofile}
Under the above hypotheses, $S_{F}$ is a $(q+6,3)$-arc with
\[
(i,t,e)=\Bigl(\tfrac{1+\eta(-1)}{2}+\tfrac{1-\eta(\gamma^{2}-\beta)}{2},\ 2,\ 2\Bigr),
\]
so that
\[
A_{q+3}=
\begin{cases}
\dfrac{(5q+7)(q-1)}{2}, & q\equiv 1\ (\mathrm{mod}\ 4),\ \eta(\gamma^{2}-\beta)=-1,\\[6pt]
\dfrac{(5q+5)(q-1)}{2}, & q\equiv 1\ (\mathrm{mod}\ 4),\ \eta(\gamma^{2}-\beta)=1,\\[6pt]
\dfrac{(5q+5)(q-1)}{2}, & q\equiv 3\ (\mathrm{mod}\ 4),\ \eta(\gamma^{2}-\beta)=-1,\\[6pt]
\dfrac{(5q+3)(q-1)}{2}, & q\equiv 3\ (\mathrm{mod}\ 4),\ \eta(\gamma^{2}-\beta)=1,
\end{cases}
\]
in accordance with \cite[Theorem 3]{FWX24}.
\end{proposition}

\begin{proof}
The ten pairs of $T_{F}$ lie on the following six lines: the tangents $L_{\infty}$
(carrying $(1),(0)$) and $l_{0,0}$ (carrying $(0),(0,-\beta)$); the external lines
$l_{0,\beta}$ (carrying $(0),(\beta,0),(\beta,\gamma)$; external as $\eta(\beta)=-1$) and
$l_{1,\beta}$ (carrying $(1),(\beta,0),(0,-\beta)$; external as $\eta(1+4\beta)=-1$); the line
$l_{1,\beta-\gamma}$ (carrying $(1),(\beta,\gamma)$; external as $\eta(1+4\beta-4\gamma)=-1$);
and the line joining $(0,-\beta)$ and $(\beta,\gamma)$, which is $l_{a,a\beta}$ with
$a=\beta/(\beta+\gamma)$ (note $\beta+\gamma\neq0$, since $\gamma=-\beta$ would force
$\eta(1+4\beta+4\gamma)=\eta(1)=1$) and discriminant
$a^{2}+4a\beta=a(a+4\beta)$, of quadratic character
\[
\eta(a)\,\eta(a+4\beta)
=\bigl[\eta(\beta)\eta(\beta+\gamma)\bigr]\cdot
\bigl[\eta(\beta)\eta(1+4\beta+4\gamma)\eta(\beta+\gamma)\bigr]=-1,
\]
so this line is external as well. That $S_{F}$ is a $(q+6,3)$-arc is proved in
\cite[Theorem 3]{FWX24};
since every pair of $T_{F}$ lies on one of the six lines above, of which exactly two are
tangents carrying two points of $T_{F}$ and exactly two are external lines carrying three, we
get $t=2$ and $e=2$. The points $(1),(0),(0,-\beta)$ lie on tangents, hence are
exterior. By \eqref{eq:inout}, applied to $(\beta,0)$,
$\eta(0^{2}-\beta)=\eta(-1)\eta(\beta)=-\eta(-1)$, so $(\beta,0)$ is interior precisely when
$q\equiv1\pmod4$; and
$(\beta,\gamma)$ is interior precisely when $\eta(\gamma^{2}-\beta)=-1$. Lemma~\ref{lem:identity}
and Lemma~\ref{lem:DL} give the displayed values.
\end{proof}

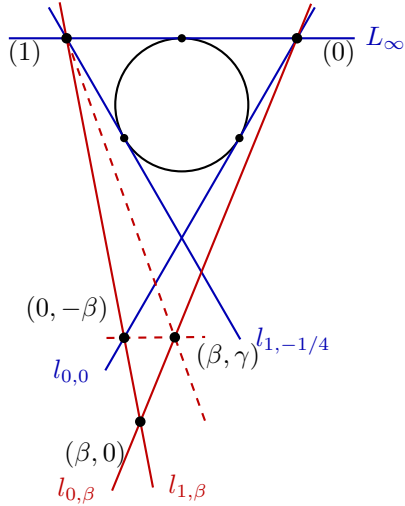
\begin{figure}[htb]
\centering
\begin{tikzpicture}[scale=0.88,
  tang/.style={blue!70!black, thick},
  extl/.style={red!75!black, thick},
  pt/.style={circle, fill=black, inner sep=1.4pt},
  cpt/.style={circle, fill=black, inner sep=1.1pt},
  lb/.style={font=\small}]
\draw[thick] (0,0) circle (1);
\draw[tang] (-2.6,1) -- (2.6,1) node[lb, right] {$L_{\infty}$};
\draw[tang] (2.0,1.464) -- (-1.15,-3.99);             
\node[lb, tang, anchor=east] at (-1.22,-4.02) {$l_{0,0}$};
\draw[tang] (-2.0,1.464) -- (0.88,-3.524);            
\node[lb, tang, anchor=west] at (0.95,-3.55) {$l_{1,-1/4}$};
\draw[extl] (-1.836,1.540) -- (-0.433,-5.750);        
\node[lb, extl, anchor=west] at (-0.36,-5.75) {$l_{1,\beta}$};
\draw[extl] (1.944,1.518) -- (-1.048,-5.797);         
\node[lb, extl, anchor=east] at (-1.12,-5.80) {$l_{0,\beta}$};
\draw[extl, dashed] (-1.132,-3.503) -- (0.351,-3.488);
\draw[extl, dashed] (-1.732,1) -- (0.383,-4.841);     
\node[pt] at (-1.732,1) {}; \node[lb, anchor=east] at (-1.95,0.78) {$(1)$};
\node[pt] at (1.732,1) {};  \node[lb, anchor=west] at (1.95,0.78) {$(0)$};
\node[pt] at (-0.866,-3.5) {};  \node[lb, anchor=south east] at (-0.95,-3.42) {$(0,-\beta)$};
\node[pt] at (-0.624,-4.76) {}; \node[lb, anchor=north east] at (-0.68,-4.90) {$(\beta,0)$};
\node[pt] at (-0.105,-3.493) {}; \node[lb, anchor=west] at (0.06,-3.78) {$(\beta,\gamma)$};
\node[cpt] at (0,1) {};
\node[cpt] at (0.866,-0.5) {};
\node[cpt] at (-0.866,-0.5) {};
\end{tikzpicture}
\caption{The Fan--Wang--Xu configuration $S_{F}$ of Proposition~\ref{prop:fwxprofile}, drawn
in the same projective frame as Figure~\ref{fig:config}. Only the two tangents $L_{\infty}$
and $l_{0,0}$ carry two points of $T_{F}$, so $t=2$; the external lines $l_{1,\beta}$ and
$l_{0,\beta}$ carry three points of $T_{F}$ each ($e=2$), and the two dashed external lines
carry two points each. The tangent $l_{1,-1/4}$ contains no point of $T_{F}$, in contrast with
Figure~\ref{fig:config}; this is the geometric source of the separation in
Theorem~\ref{thm:fullineq}.}
\label{fig:fwx}
\end{figure}

Since monomially equivalent codes have equal weight enumerators, comparing the value lists of
Theorem~\ref{thm:main} and Proposition~\ref{prop:fwxprofile} yields:

\begin{corollary}\label{cor:ineq}
Let $q$ be an odd prime power and let $\Code_{(b_{1},b_{2})}$ be as in Theorem~\ref{thm:main}.
\begin{enumerate}
\item If $\Delta=0$ and $\eta(\delta)=-1$, then $A_{q+3}=(5q+9)(q-1)/2$ and
$\Code_{(b_{1},b_{2})}$ is monomially inequivalent to every code of the family of \cite{FWX24}.
\item If $q\equiv 3\pmod 4$ and $A_{q+3}=(5q+7)(q-1)/2$ (equivalently, by
Proposition~\ref{prop:real}, $\eta(\Delta)=-1$ and $\eta(\delta)=-1$), then $\Code_{(b_{1},b_{2})}$ is monomially inequivalent to
every code of the family of \cite{FWX24}.
\end{enumerate}
\end{corollary}

By Proposition~\ref{prop:real}, the hypothesis of Corollary~\ref{cor:ineq}(1) is satisfiable
precisely when $q\equiv3\pmod4$; and for $q\equiv1\pmod4$ the sets of values of $A_{q+3}$
realized by the two families coincide, so weight enumerators alone cannot separate the
families (finer counting invariants such as the complete weight enumerator can distinguish
particular pairs of codes, as in \cite[Remark 3.20]{QDY26}, but do not settle the question in
general; see the framework of \cite{LZ25}). Nevertheless, the
triple $(i,t,e)$ itself is an invariant of monomial equivalence for $q\ge9$, and it separates
the two families completely. We first record two standard facts.

\begin{lemma}\label{lem:monomial}
Let $\Code_{1},\Code_{2}$ be \emph{projective} $[n,3]$ codes over $\Fq$, i.e.\ codes whose
generator matrices have pairwise nonproportional nonzero columns, with associated point sets
$S_{1},S_{2}\subseteq\PG(2,q)$.
Then $\Code_{1}$ and $\Code_{2}$ are monomially equivalent if and only if some element of
$\mathrm{PGL}(3,q)$ maps $S_{1}$ onto $S_{2}$.
\end{lemma}

\begin{proof}
$\Code_{2}=\{cM: c\in\Code_{1}\}$ for an $n\times n$ monomial matrix $M$ (Section~\ref{sec:pre}) if and only if $G_{2}=AG_{1}M$ for some
$A\in\mathrm{GL}(3,q)$, i.e.\ if and only if the columns of $G_{2}$ are, up to order and
nonzero scalars, the images of the columns of $G_{1}$ under $A$; projectively this says
exactly that the collineation induced by $A$ maps $S_{1}$ onto $S_{2}$.
\end{proof}

\begin{lemma}\label{lem:rigid}
Let $q\ge m+4$ and, for $j=1,2$, let $S_{j}=\mathcal{O}_{j}\cup T_{j}$ be a $(q+1+m,3)$-arc,
where $\mathcal{O}_{j}$ is a conic and $\#T_{j}=m$. If $g\in\mathrm{PGL}(3,q)$ maps $S_{1}$
onto $S_{2}$, then $g(\mathcal{O}_{1})=\mathcal{O}_{2}$ and $g(T_{1})=T_{2}$; consequently the
triples $(i,t,e)$ of Lemma~\ref{lem:identity} associated with $S_{1}$ and $S_{2}$ coincide.
\end{lemma}

\begin{proof}
The points of $g(\mathcal{O}_{1})$ not lying on $\mathcal{O}_{2}$ lie in $T_{2}$, so
$\#\bigl(g(\mathcal{O}_{1})\cap\mathcal{O}_{2}\bigr)\ge q+1-m\ge5$. Since five points, no three
collinear, lie on a unique conic \cite[Corollary to Theorem~7.2.1]{Hirschfeld},
$g(\mathcal{O}_{1})=\mathcal{O}_{2}$, and hence
$g(T_{1})=g(S_{1})\setminus g(\mathcal{O}_{1})=T_{2}$. A collineation mapping
$\mathcal{O}_{1}$ onto $\mathcal{O}_{2}$ maps tangent, secant and external lines of
$\mathcal{O}_{1}$ to lines of the same type with respect to $\mathcal{O}_{2}$ (the type is
determined by $\#(\ell\cap\mathcal{O}_{j})$), and therefore maps exterior and interior points
to points of the same type. The quantities $i$, $t$, $e$ are defined purely in terms of these
notions and of $T_{j}$, so they are preserved.
\end{proof}

\begin{theorem}\label{thm:fullineq}
Let $q\ge9$ be an odd prime power. Every code $\Code_{(b_{1},b_{2})}$ of Theorem~\ref{thm:main} is
monomially inequivalent to every code of the family of \cite{FWX24}.
\end{theorem}

\begin{proof}
By the proof of Theorem~\ref{thm:main}, the arc of $\Code_{(b_{1},b_{2})}$ has $t\in\{3,4\}$,
whereas by Proposition~\ref{prop:fwxprofile} the arc of any code of \cite{FWX24} has $t=2$. By
Lemmas \ref{lem:monomial} and \ref{lem:rigid} (with $m=5$, so $q\ge9$ suffices), monomially
equivalent codes of this shape have equal $t$; the claim follows.
\end{proof}

\begin{remark}\label{rem:smallq}
For $q\le9$ Theorem~\ref{thm:fullineq} does not apply, but no comparison is needed there: by
\cite[Lemma 1]{FWX24}, for each admissible $\beta$ the number of admissible $\gamma$ equals
$\bigl(q-6-\eta(-1)-2\eta(2)\bigr)/4$, which vanishes precisely for $q\in\{3,5,7,9\}$ (as
$q-6\le\eta(-1)+2\eta(2)$ forces $q\le9$). Hence the family of \cite{FWX24} is empty for
$q\le9$ and nonempty for every odd prime power $q\ge11$, and no code of Theorem~\ref{thm:main}
is monomially equivalent to a code of \cite{FWX24} for \emph{any} odd prime power $q$.
\end{remark}

We record Magma-verified \cite{Magma} examples for the three enumerators; in each case the
weight enumerator coincides with Theorem~\ref{thm:main}.

\begin{example}\label{ex:q3}
$q=3$, $b_{1}=2$, $b_{2}=1$: then $\Delta=0$ and $\eta(\delta)=\eta(-1)=-1$, so
$A_{6}=24$ and $A(y)=1+24y^{6}+2y^{9}$. By Corollary~\ref{cor:ineq}(1) this $[9,3,6]$ code is
inequivalent to the codes of \cite{FWX24}.
\end{example}

\begin{example}\label{ex:q5}
$q=5$, $b_{1}=2$, $b_{2}=3$: then $\Delta=0$ and $\eta(\delta)=\eta(4)=1$, so
$A_{8}=64$ and $A(y)=1+64y^{8}+28y^{9}+16y^{10}+16y^{11}$.
\end{example}

\begin{example}\label{ex:q7}
$q=7$. For $b_{1}=3$, $b_{2}=3$ we have $\eta(\Delta)=\eta(3)=-1$ and $\eta(\delta)=\eta(4)=1$,
so $A_{10}=120$ and $A(y)=1+120y^{10}+108y^{11}+48y^{12}+66y^{13}$. For $b_{1}=6$, $b_{2}=4$ we
have $\eta(\Delta)=\eta(5)=-1$ and $\eta(\delta)=\eta(5)=-1$, so $A_{10}=126$ and
$A(y)=1+126y^{10}+90y^{11}+66y^{12}+60y^{13}$; since $7\equiv3\pmod4$, this $[13,3,10]$ code is
inequivalent to the codes of \cite{FWX24} by Corollary~\ref{cor:ineq}(2). For $b_{1}=3$,
$b_{2}=4$ we have $\Delta=0$ and $\eta(\delta)=\eta(6)=-1$, so $A_{10}=132$ and
$A(y)=1+132y^{10}+72y^{11}+84y^{12}+54y^{13}$; this code is inequivalent to the codes of
\cite{FWX24} by Corollary~\ref{cor:ineq}(1).
\end{example}

\section{A sixth point: $[q+7,3,q+4]$ NMDS codes}\label{sec:ext}

Throughout this section let $q$ be odd and let $(b_{1},b_{2})$ satisfy the hypotheses of
Theorem~\ref{thm:main} with $\eta(\Delta)=-1$, so that $l_{A,B}$ is an external line carrying
the two points $P_{3},P_{5}$ of $T$. We seek a sixth point on $l_{A,B}$: for
$v_{6}\in\Fq$ put
\[
P_{6}=(u_{6},v_{6}),\qquad u_{6}:=Av_{6}+B,
\qquad A=\frac{b_{1}}{\,b_{1}+b_{2}+\tfrac14\,},\quad B=Ab_{2}.
\]
Since $l_{A,B}$ is external, $P_{6}\notin\mathcal{O}$ automatically. Writing
$v_{4}=b_{1}-b_{2}$ for the second coordinate of $P_{4}$, and recalling
$\delta=(b_{1}-b_{2})^{2}-b_{1}$ from \eqref{eq:Delta}, define
\begin{equation}\label{eq:g}
g(v_{6}):=(u_{6}-b_{1})^{2}
+4\bigl(b_{1}v_{6}-(b_{1}-b_{2})u_{6}\bigr)\bigl(v_{6}-v_{4}\bigr),
\end{equation}
which is, up to the square factor $(v_{6}-v_{4})^{2}$, the discriminant of the line joining
$P_{6}$ and $P_{4}$; as $v_{6}\neq v_{4}$ for the points considered below, this factor is a
nonzero square, so $\eta(g(v_{6}))$ equals the quadratic character of that discriminant.

\begin{theorem}\label{thm:q7}
Let $q$ be an odd prime power, let $(b_{1},b_{2})$ satisfy the hypotheses of
Theorem~\ref{thm:main} with $\eta(\Delta)=-1$, and let
$v_{6}\in\Fq\setminus\bigl\{b_{1}+\tfrac14,\,-b_{2},\,b_{1}-b_{2}\bigr\}$ satisfy
\[
\eta(u_{6})=-1,\qquad
\eta\bigl(1+4(u_{6}-v_{6})\bigr)=-1,\qquad
g(v_{6})=0\ \text{ or }\ \eta\bigl(g(v_{6})\bigr)=-1 .
\]
Then $S\cup\{P_{6}\}$ is a $(q+7,3)$-arc, and the code $\Code'$ generated by
$\bigl(G_{(b_{1},b_{2})}\mid P_{6}\bigr)$ is a $[q+7,3,q+4]$ NMDS code over $\Fq$ with
\[
A_{q+4}=\Bigl(3q+3+\tfrac{1-\eta(\delta)}{2}+\tfrac{1-\eta(v_{6}^{2}-u_{6})}{2}
+[\,g(v_{6})=0\,]\Bigr)(q-1),
\]
where moreover $g(v_{6})=0$ forces $\eta(\delta)=\eta(v_{6}^{2}-u_{6})=1$; consequently
$A_{q+4}/(q-1)\in\{3q+3,\;3q+4,\;3q+5\}$. The complete weight distribution follows from
Lemma~\ref{lem:DL}.
\end{theorem}

\begin{proof}
Write $T'=T\cup\{P_{6}\}$. First, $P_{6}$ is distinct from $P_{1},\dots,P_{5}$: it is affine
with $u_{6}\neq0$ (so $P_{6}\neq P_{1},P_{2},P_{5}$), its second coordinate avoids those of
$P_{3},P_{5}$, and $P_{4}\notin l_{A,B}$ by Step~3 of the proof of Theorem~\ref{thm:main}.

We check that every pair of points of $T'$ lies on a tangent or external line, and that no line
carries four points of $S\cup\{P_{6}\}$. The pairs within $T$ are handled by
Theorem~\ref{thm:main}, and none of the six distinguished lines of its proof acquires $P_{6}$
except $l_{A,B}$: indeed $P_{6}\in L_{\infty}$ is impossible ($P_{6}$ affine);
$P_{6}\in l_{1,-1/4}$ or $P_{6}\in l_{0,b_{1}}$ would force $P_{6}=l_{A,B}\cap l_{1,-1/4}=P_{3}$
or $P_{6}=l_{A,B}\cap l_{0,b_{1}}=P_{3}$; $P_{6}\in l_{0,0}$ would force $u_{6}=0$; and
$P_{6}\in l_{1,b_{2}}$ would force $P_{6}=l_{A,B}\cap l_{1,b_{2}}=P_{5}$. (Each of these
intersections is $P_{3}$ or $P_{5}$ because the two lines involved already share that point.)
Thus $l_{A,B}$ becomes an external line carrying the three points $P_{3},P_{5},P_{6}$, which is
permitted, and the pairs $\{P_{6},P_{3}\}$, $\{P_{6},P_{5}\}$ are covered.

For the remaining three pairs, the hypotheses on $v_{6}$ do the work. The line $P_{6}P_{1}$
is $l_{1,\,u_{6}-v_{6}}$, which the hypothesis $\eta(1+4(u_{6}-v_{6}))=-1$ makes external;
this is precisely what the configuration needs, since the two tangents through the exterior
point $P_{1}$, namely $L_{\infty}$ and $l_{1,-1/4}$, already carry two points of $T$ each,
while a secant $P_{6}P_{1}$ would carry four points of $S'$. Likewise $P_{6}P_{2}=l_{0,u_{6}}$ is external by
$\eta(u_{6})=-1$. For $P_{6}P_{4}$: the exclusion $v_{6}\neq b_{1}-b_{2}$ rules out the secant
$L_{v_{6}}$, so $P_{6}P_{4}=l_{a,c}$ with $a=(u_{6}-b_{1})/(v_{6}-v_{4})$ and
$c=u_{6}-av_{6}$, and $a^{2}+4c=g(v_{6})/(v_{6}-v_{4})^{2}$; the hypothesis on $g$ makes this
line tangent or external. Finally, none of these three lines carries a further point of $T$:
membership of any $P_{j}$ ($j\le5$) would identify the line with one of the six distinguished
lines and hence force $P_{6}\in\{P_{3},P_{5}\}$, as above. Therefore $S\cup\{P_{6}\}$ is a
$(q+7,3)$-arc and $\Code'$ is a $[q+7,3,q+4]$ NMDS code.

For the weight distribution we apply Lemma~\ref{lem:identity} with $m=6$. Every pair of $T'$
lies on a unique line, and these lines are exactly: the tangents $L_{\infty}$, $l_{1,-1/4}$,
$l_{0,0}$ (two points each), the external lines $l_{0,b_{1}}$, $l_{1,b_{2}}$, $l_{A,B}$ (three
points each), and the three lines $P_{6}P_{1}$, $P_{6}P_{2}$, $P_{6}P_{4}$ (two points each,
external except possibly $P_{6}P_{4}$). Hence $e=3$, $t=3+[\,g(v_{6})=0\,]$, and
$i=\frac{1-\eta(\delta)}{2}+\frac{1-\eta(v_{6}^{2}-u_{6})}{2}$ by \eqref{eq:inout}. If
$g(v_{6})=0$ then $P_{6}P_{4}$ is a tangent through both $P_{4}$ and $P_{6}$, so both points
are exterior and the two character terms vanish. Lemma~\ref{lem:identity} and \eqref{eq:Amin}
give the stated $A_{q+4}$, and Lemma~\ref{lem:DL} the full distribution.
\end{proof}

\begin{proposition}\label{prop:q7exist}
Admissible triples $(b_{1},b_{2},v_{6})$ in Theorem~\ref{thm:q7} exist for every odd prime
power $q\ge11$, and for no odd $q\le9$.
\end{proposition}

\begin{proof}
For $q\le9$ nonexistence was verified by exhaustive computation, as was existence for the
$31$ odd prime powers $11\le q\le121$ (for $11\le q\le27$ the numbers of admissible triples are
$8,\,12,\,24,\,64,\,100,\,124,\,228$). Let now $q\ge125$ (so that $q-7\sqrt q-46>0$).

Both existence proofs below use the same standard device: for $x\neq0$ the factor $1-\eta(x)$
equals $2$ if $x$ is a nonsquare and $0$ if it is a square (similarly for $1+\eta(x)$), so a
product of such factors equals $2^{k}$ at the elements satisfying all $k$ character conditions
and $0$ at every other element where no argument vanishes. Summing over $\Fq$ and expanding,
the resulting character sums are evaluated by Lemma~\ref{lem:weil}: exactly in degrees one and
two, and by the Weil bound in degrees three and four. Elements at which some argument vanishes
contribute intermediate values and are accounted for separately.

Pairs $(b_{1},b_{2})$ satisfying the hypotheses of Theorem~\ref{thm:main} with
$\eta(\Delta)=-1$ exist: fixing any nonsquare $b_{1}$, the conditions
on $b_{2}$ read $\eta(h_{1})=-1$ and $\eta(h_{2})=1$, where $h_{1}=1+4b_{2}$ and
$h_{2}=(b_{2}+\tfrac14)(b_{1}+b_{2})$: the first is the hypothesis $\eta(1+4b_{2})=-1$
itself, and the second is equivalent to $\eta(\Delta)=-1$ by \eqref{eq:Deltafact}, since
$\eta(\Delta)=\eta(4b_{1})\,\eta(h_{2})=-\eta(h_{2})$ because $\eta(b_{1})=-1$. The
quadratic $h_{2}$ is squarefree (its roots $-\tfrac14$ and $-b_{1}$ differ since
$b_{1}\neq\tfrac14$, a square). Since $h_{1}=4\bigl(b_{2}+\tfrac14\bigr)$, we have
$h_{1}h_{2}=4\bigl(b_{2}+\tfrac14\bigr)^{2}(b_{1}+b_{2})$, so
$\eta(h_{1}h_{2})=\eta(b_{1}+b_{2})$ for $b_{2}\neq-\tfrac14$ and hence
$\sum_{b_{2}\in\Fq}\eta(h_{1}h_{2})=\sum_{b_{2}\neq-\tfrac14}\eta(b_{1}+b_{2})
=-\eta\bigl(b_{1}-\tfrac14\bigr)$. The exact evaluations of Lemma~\ref{lem:weil} give
$\sum_{b_{2}}\eta(h_{1})=0$ (linear) and $\sum_{b_{2}}\eta(h_{2})=-\eta(1)=-1$ (squarefree
quadratic); together with the value $-\eta(b_{1}-\tfrac14)\ge-1$ just computed, whatever the
sign of $\eta(b_{1}-\tfrac14)$,
\[
\sum_{b_{2}\in\Fq}\bigl(1-\eta(h_{1})\bigr)\bigl(1+\eta(h_{2})\bigr)
= q-\sum\eta(h_{1})+\sum\eta(h_{2})-\sum\eta(h_{1}h_{2})
\ \ge\ q-0-1-1\ =\ q-2 .
\]
Each admissible $b_{2}$ contributes $4$ to the sum; the two zeros
$b_{2}\in\bigl\{-\tfrac14,\,-b_{1}\bigr\}$ of $h_{1}h_{2}$ contribute at most $1$ and $2$
respectively (at $b_{2}=-\tfrac14$ both characters vanish and the term equals $(1-0)(1+0)=1$;
at $b_{2}=-b_{1}$ the term is $1-\eta(1-4b_{1})\le2$), and every other $b_{2}$ contributes
$0$. Hence, writing $N$ for the number of admissible $b_{2}$ distinct from the excluded value
$b_{2}=-b_{1}-\tfrac14$ (which itself contributes at most $4$), the sum is at most
$4N+1+2+4$; combined with the lower bound $q-2$ this gives $4N\ge q-9$, so
$N\ge\tfrac14\bigl(q-9\bigr)>0$.

Fix such a pair and regard the three quantities of Theorem~\ref{thm:q7} as polynomials in
$v=v_{6}$. Their coefficients involve the denominator $s=b_{1}+b_{2}+\tfrac14$, so we display
the versions rescaled by $s$ and $s^{2}$ (rescaling by a nonzero constant multiplies every
value of $\eta$ by a fixed sign, so it affects none of the evaluations below):
\[
sf_{1}(v)=b_{1}(v+b_{2}),\qquad
sf_{2}(v)=\tfrac14(4b_{2}+1)\bigl(4b_{1}+1-4v\bigr),\qquad
s^{2}g(v),
\]
where $f_{1}(v)=Av+B$ and $f_{2}(v)=1+4\bigl((A-1)v+B\bigr)$ are the quantities $u_{6}$ and
$1+4(u_{6}-v_{6})$ of Theorem~\ref{thm:q7} regarded as functions of $v=v_{6}$, and $g$ is as
in \eqref{eq:g}. Then $f_{1}$ and
$f_{2}$ are nonproportional linear polynomials (their roots $-b_{2}$ and $b_{1}+\tfrac14$
coincide only if $s=0$, which is excluded),
and direct computation (\ref{app:identities}) gives, writing
$s^{2}g=c_{2}v^{2}+c_{1}v+c_{0}$,
\[
c_{2}\;=\;\tfrac14\,b_{1}(4b_{2}+1)\bigl(8(b_{1}+b_{2})+1\bigr),
\qquad
\operatorname{disc}_{v}\bigl(s^{2}g\bigr)\;=\;\bigl(b_{1}(4b_{2}+1)\,s\bigr)^{2}\,\delta .
\]
Since $\delta\neq0$ (recall from Section~\ref{sec:main} that $\eta(b_{1})=-1$ forces
$\delta\neq0$), whenever $\deg g=2$ the polynomial $g$ is squarefree; and, as $b_{1}\neq0$
and $4b_{2}+1\neq0$, $\deg g<2$ occurs only for $b_{1}+b_{2}=-\tfrac18$. In that case
$c_{1}=-\tfrac14\,b_{1}(4b_{2}+1)\bigl(b_{1}-\tfrac1{16}\bigr)\neq0$
($b_{1}\notin\{0,\tfrac1{16}\}$, being a nonsquare), so $g$ is a nonconstant linear
polynomial; moreover $g$ is then proportional to neither $f_{1}$ nor $f_{2}$, because
$s^{2}g(-b_{2})=-\tfrac1{128}\,b_{1}^{2}(8b_{1}-1)$ and
$s^{2}g\bigl(b_{1}+\tfrac14\bigr)=\tfrac1{1024}\,b_{1}(8b_{1}-1)^{2}$ are nonzero
($b_{1}=\tfrac18$ would force $4b_{2}+1=0$). Finally $g$ is never proportional to
$f_{1}f_{2}$: the proportionality of coefficients forces $b_{2}=-\tfrac14$. The displayed
identities and the last three assertions are verified in \ref{app:identities}.

Consequently each of the seven products of distinct elements of $\{f_{1},f_{2},g\}$ has
squarefree part of degree between $1$ and $4$: this is clear when the three polynomials are
pairwise coprime, and if $f_{1}$ or $f_{2}$ divides the quadratic $g$, the square factor drops
out of the product and leaves a nonconstant squarefree part, since $g$ is squarefree and not
proportional to $f_{1}f_{2}$. Writing $\Sigma_{h}:=\sum_{v\in\Fq}\eta\bigl(h(v)\bigr)$ for a
polynomial $h$, the indicator product expands as
\[
\sum_{v\in\Fq}\bigl(1-\eta(f_{1})\bigr)\bigl(1-\eta(f_{2})\bigr)\bigl(1-\eta(g)\bigr)
= q-\Sigma_{f_{1}}-\Sigma_{f_{2}}-\Sigma_{g}
+\Sigma_{f_{1}f_{2}}+\Sigma_{f_{1}g}+\Sigma_{f_{2}g}-\Sigma_{f_{1}f_{2}g}.
\]
The sums $\Sigma_{f_{1}}$, $\Sigma_{f_{2}}$, $\Sigma_{g}$ and $\Sigma_{f_{1}f_{2}}$ are again
evaluated exactly by Lemma~\ref{lem:weil}: the two linear sums vanish and the two quadratic
sums have absolute value at most $1$, so $-\Sigma_{g}+\Sigma_{f_{1}f_{2}}\ge-2$, which yields
the term $-2$ below. The sums $\Sigma_{f_{1}g}$, $\Sigma_{f_{2}g}$ and
$\Sigma_{f_{1}f_{2}g}$ are estimated by the Weil
bound applied to their squarefree parts (at most $2\sqrt q$, $2\sqrt q$ and $3\sqrt q$ in
absolute value); the additive constants in $2\sqrt q+1$ and $3\sqrt q+2$ cover the
degenerate cases where $f_{1}$ or $f_{2}$ divides $g$: there discarding the square factor
changes the sum by at most $1$, at the zero of the discarded factor, while the remaining
squarefree sum has degree at most $2$ and is evaluated exactly by Lemma~\ref{lem:weil},
with absolute value at most $1$; hence
\[
\sum_{v\in\Fq}\bigl(1-\eta(f_{1})\bigr)\bigl(1-\eta(f_{2})\bigr)\bigl(1-\eta(g)\bigr)
\;\ge\; q-2-2\bigl(2\sqrt q+1\bigr)-\bigl(3\sqrt q+2\bigr)\;=\;q-7\sqrt q-6 .
\]
On the other hand, each $v$ with $\eta(f_{1})=\eta(f_{2})=\eta(g)=-1$ contributes $8$ to the
left-hand side, each of the at most four zeros of $f_{1}f_{2}g$ contributes at most $4$ (there one factor of
the indicator product equals $1$ and the other two are at most $2$ each), and
every other $v$ contributes $0$; hence the left-hand side is at most $8N'+16$, where $N'$ is
the number of $v$ of the first kind (note that every such $v$, and also every zero of $g$ at
which $\eta(f_{1})=\eta(f_{2})=-1$, is admissible for Theorem~\ref{thm:q7}). Combining the
two bounds, $8N'+16\ge q-7\sqrt q-6$, that is,
$N'\ge\tfrac18\bigl(q-7\sqrt q-22\bigr)$. Discarding the three excluded values of $v_{6}$,
each of which may be among the $v$ counted by $N'$, the number of admissible $v_{6}$ is at
least $\tfrac18\bigl(q-7\sqrt q-22\bigr)-3=\tfrac18\bigl(q-7\sqrt q-46\bigr)>0$.
\end{proof}

\begin{corollary}\label{cor:q7}
For every odd prime power $q\ge11$ there exists a $[q+7,3,q+4]$ NMDS code over $\Fq$ as in
Theorem~\ref{thm:q7}.
\end{corollary}

We close this section by recording the locality of all codes constructed in this paper, a
property of interest for locally recoverable codes (LRCs); see \cite{CM15,LH23,ZDQ26}. The
$i$-th coordinate of a code $\Code$ has \emph{locality} $r$ if there is a codeword of
$\Code^{\perp}$ of weight at most $r+1$ whose support contains $i$ (so that the $i$-th symbol of
any codeword can be recovered from at most $r$ other symbols), and $\Code$ has locality $r$ if
every coordinate does.

\begin{proposition}\label{prop:lrc}
Every code of Theorem~\ref{thm:main} and of Theorem~\ref{thm:q7} has locality $2$ and is an
optimal locally recoverable code with respect to the Cadambe--Mazumdar bound \cite{CM15}.
\end{proposition}

\begin{proof}
The coordinates of $\Code$ correspond to the points of the arc $S$, and the weight-$3$ codewords
of $\Code^{\perp}$ correspond, up to scalars, to the trisecants of $S$: three columns of the
generator matrix are linearly dependent if and only if the corresponding points are collinear,
and the support of the resulting dual codeword consists of the three corresponding
coordinates. Hence a coordinate has locality $2$ if and only if its point lies on a trisecant
of $S$. For $P\in T$ this holds because every secant of $\mathcal{O}$ through $P$ is a
trisecant of $S$ and there are at least $\tfrac{q-1}{2}\ge1$ of them. For $P=O_{i}\in
\mathcal{O}$, note that the tangent of $\mathcal{O}$ at $O_{i}$ carries at most two points of
$T$; since $\#T\ge5$, there is $P_{j}\in T$ such that the line $O_{i}P_{j}$ is not that
tangent. A non-tangent line through $O_{i}$ is a secant of $\mathcal{O}$, and by the arc
condition it carries exactly three points of $S$, namely $O_{i}$, a second point of
$\mathcal{O}$, and $P_{j}$. Hence $\Code$ has locality $2$.

For the optimality, the Cadambe--Mazumdar bound \cite{CM15} states that a $q$-ary $[n,k,d]$
code with locality $r$ satisfies
$k\le\min_{\mu\ge1}\bigl[\,\mu r+k^{(q)}_{\mathrm{opt}}\bigl(n-\mu(r+1),d\bigr)\bigr]$, the
minimum being over integers $\mu\ge1$, where
$k^{(q)}_{\mathrm{opt}}(n',d)$ is the largest dimension of a $q$-ary code of length $n'$ and
minimum distance $d$. Taking $r=2$ and $\mu=1$, the Singleton bound gives
$k^{(q)}_{\mathrm{opt}}(n-3,\,n-3)\le1$, so $k\le3$ for every code of length $n$, minimum
distance $n-3$ and locality $2$; our codes attain this with $k=3$.
\end{proof}

\begin{example}
Let $q=11$, $b_{1}=2$, $b_{2}=4$, $v_{6}=1$. Then $\eta(\Delta)=\eta(6)=-1$,
$\eta(\delta)=\eta(2)=-1$, $u_{6}=6$, $\eta(u_{6})=-1$, $\eta(1+4(u_{6}-v_{6}))=\eta(10)=-1$,
$g(1)=8$ with $\eta(8)=-1$, and $\eta(v_{6}^{2}-u_{6})=\eta(6)=-1$; thus both $P_{4}$ and
$P_{6}$ are interior and $A_{15}=(3q+5)(q-1)=380$. Direct enumeration of all $q^{3}-1$
codewords confirms that $\Code'$ is an $[18,3,15]$ NMDS code with
\[
A(y)=1+380y^{15}+390y^{16}+240y^{17}+320y^{18},
\]
in agreement with Theorem~\ref{thm:q7} and Lemma~\ref{lem:DL}.
\end{example}

\section{Concluding remarks}\label{sec:conc}

For every odd prime power $q$ we constructed $[q+6,3,q+3]$ NMDS codes by extending a conic of
$\PG(2,q)$ by five points, determined their three weight enumerators through the counting
identity of Lemma~\ref{lem:identity}, and showed that for every odd prime power $q$ the whole family
is monomially inequivalent to the family of Fan, Wang and Xu \cite{FWX24} with the same
parameters, the separating invariant being the triple $(i,t,e)$. Adding a
sixth point on the external line $l_{A,B}$ we further obtained $[q+7,3,q+4]$ NMDS codes for
every odd prime power $q\ge11$, complementing in odd characteristic the length-$(q+7)$
families known for even $q$ \cite{FWX24}. All these codes are optimal locally recoverable
codes with locality $2$ (Proposition~\ref{prop:lrc}). Together
with \cite{FWX24}, the values of $A_{q+3}/(q-1)$ realized by explicit constructions are now
$\tfrac{5q+3}{2},\tfrac{5q+5}{2},\tfrac{5q+7}{2},\tfrac{5q+9}{2}$. Lemma~\ref{lem:identity}
suggests several natural problems, which we intend to address in future work.

\begin{problem}
By Lemma~\ref{lem:identity}, any $(q+6,3)$-arc containing a conic has
$A_{q+3}=\bigl(\tfrac{5(q-1)}{2}+i+t+e\bigr)(q-1)$. Exhaustive computation for $q\le 9$ shows
that the realized values of $i+t+e$ form the set $\{3,4,5,6,7\}$ for $q=7,9$ (and $\{5,6\}$
for $q=5$), but the lower bound $3$ does not persist: over $\mathbb{F}_{11}$ the five points
$(7,0)$, $(0,7)$, $(9,10)$, $(7,4)$, $(1)$ extend the conic \eqref{eq:conic} to a
$(17,3)$-arc with $i=t=e=0$, attaining $A_{q+3}=\tfrac{5}{2}(q-1)^{2}$, while over
$\mathbb{F}_{13}$ an exhaustive search shows that the minimum of $i+t+e$ equals $1$. Determine the
exact spectrum of $i+t+e$ for all $q$; in particular, decide for which $q$ the value $0$ is
attained, and construct explicit infinite families realizing the extreme values. The
analogous question for the $(q+7,3)$-arcs of Section~\ref{sec:ext}, where the identity reads
$A_{q+4}=\bigl(3(q-1)+i+t+e\bigr)(q-1)$, is equally natural.
\end{problem}

\begin{problem}\label{prob:q7}
Determine the maximal $m$, as a function of $q$, for which a conic of $\PG(2,q)$ can be
extended by $m$ points to a $(q+1+m,3)$-arc. Iterating the character-sum method of
Proposition~\ref{prop:q7exist} gives $m\gg\log q$: each further point, sought on a suitable
external line, must satisfy one quadratic-character condition for every point already added,
and the indicator-sum estimate in the proof of Proposition~\ref{prop:q7exist} remains
positive as long as $2^{m}\sqrt q$ is small compared with $q$ (Proposition~\ref{prop:q7exist}
itself carries out the step from five points to six, with three such conditions). On the
other hand $m\le q+1$
for odd $q$: the trivial bound gives $n\le 2q+3$, and equality would make the arc a maximal
$\{2q+3;3\}$-arc, which does not exist in $\PG(2,q)$ for odd $q$ \cite{BBM97}.
Conjecturally $m\le q-2$ for $q\ge8$: the NMDS length conjecture, as reported in
\cite{LZ25}, asserts that the maximal length of an $[n,3,n-3]$ NMDS code is at most $2q-1$
for $q\ge8$. Randomized computer search suggests the truth is indeed linear in
$q$: for $q=7$ and $q=9$ a conic extends to $(n,3)$-arcs of the maximal sizes $n=15=2q+1$ and
$n=17=2q-1$ \cite{MMP,MMP02} respectively, the latter attaining the conjectural bound, and
for $q=13$ to a $(22,3)$-arc.
\end{problem}

\begin{problem}
Classify the monomial equivalence classes of $[q+6,3,q+3]$ NMDS codes whose associated arcs
contain a conic. For $q\ge9$, Lemmas~\ref{lem:monomial} and~\ref{lem:rigid} show
that any monomial equivalence must map the conic to the conic, so the classes correspond to
orbits of $5$-point sets under the stabilizer $\mathrm{PGL}(2,q)$ of the conic; computer
classification gives $6$, $28$ and $119$ classes for $q=5,7,9$ respectively; the
$(q+7,3)$-arcs of Section~\ref{sec:ext} should be classified likewise. For a related
equivalence classification in even characteristic, based on regular hyperovals, see \cite{LHFZ26}.
\end{problem}

\appendix
\renewcommand{\thesection}{Appendix A}
\section{The polynomial identities of Proposition~\ref{prop:q7exist}}\label{app:identities}

We keep the notation of Section~\ref{sec:ext}: $s=b_{1}+b_{2}+\tfrac14\neq0$,
$A=b_{1}/s$, $B=Ab_{2}$, $u=Av+B$, $v_{4}=b_{1}-b_{2}$, and
$g=(u-b_{1})^{2}+4(b_{1}v-v_{4}u)(v-v_{4})$. All identities below are one-line expansions
from the following four:
\begin{align}
su&=b_{1}(v+b_{2}), \tag{A.1}\\
sf_{2}=s+4su-4sv&=\tfrac14(4b_{2}+1)(4b_{1}+1-4v), \tag{A.2}\\
s(u-b_{1})&=b_{1}\bigl(v-b_{1}-\tfrac14\bigr), \tag{A.3}\\
s(b_{1}v-v_{4}u)&=b_{1}\bigl[(2b_{2}+\tfrac14)v-(b_{1}-b_{2})b_{2}\bigr]. \tag{A.4}
\end{align}
Indeed, (A.1) is immediate; for (A.2), $s+4b_{1}(v+b_{2})-4sv
=(b_{1}+b_{2}+\tfrac14+4b_{1}b_{2})-(4b_{2}+1)v
=\tfrac14(4b_{1}+1)(4b_{2}+1)-(4b_{2}+1)v$; (A.3) is $b_{1}(v+b_{2})-b_{1}s$; and (A.4) is
$b_{1}sv-v_{4}b_{1}(v+b_{2})=b_{1}[(s-v_{4})v-v_{4}b_{2}]$ with $s-v_{4}=2b_{2}+\tfrac14$.

Multiplying $g$ by $s^{2}$ and using (A.3), (A.4),
\[
s^{2}g=b_{1}^{2}\bigl(v-b_{1}-\tfrac14\bigr)^{2}
+4s\,b_{1}\bigl[(2b_{2}+\tfrac14)v-(b_{1}-b_{2})b_{2}\bigr](v-v_{4}),
\]
whence the coefficients of $s^{2}g=c_{2}v^{2}+c_{1}v+c_{0}$ are
\begin{align}
c_{2}&=b_{1}\bigl[b_{1}+(8b_{2}+1)s\bigr], \tag{A.5}\\
c_{1}&=-b_{1}\bigl[2b_{1}\bigl(b_{1}+\tfrac14\bigr)
+4s(b_{1}-b_{2})\bigl(3b_{2}+\tfrac14\bigr)\bigr], \tag{A.6}\\
c_{0}&=b_{1}\bigl[b_{1}\bigl(b_{1}+\tfrac14\bigr)^{2}+4sb_{2}(b_{1}-b_{2})^{2}\bigr].
\tag{A.7}
\end{align}
Expanding $b_{1}+(8b_{2}+1)s=2b_{1}+8b_{1}b_{2}+8b_{2}^{2}+3b_{2}+\tfrac14
=\tfrac14(4b_{2}+1)\bigl(8(b_{1}+b_{2})+1\bigr)$ gives the leading coefficient stated in
Proposition~\ref{prop:q7exist}; in particular $c_{2}=0$ if and only if
$b_{1}+b_{2}=-\tfrac18$, and substituting $b_{2}=-\tfrac18-b_{1}$, $s=\tfrac18$ into (A.6)
yields, after simplification,
$c_{1}=b_{1}\bigl(b_{1}-\tfrac18\bigr)\bigl(b_{1}-\tfrac1{16}\bigr)
=-\tfrac14 b_{1}(4b_{2}+1)\bigl(b_{1}-\tfrac1{16}\bigr)\neq0$, since $b_{1}=\tfrac18$ would
force $4b_{2}+1=0$ and $\tfrac1{16}$ is a square.

Substituting (A.5)--(A.7) into $c_{1}^{2}-4c_{2}c_{0}$ and expanding, one obtains
\[
\operatorname{disc}_{v}\bigl(s^{2}g\bigr)=c_{1}^{2}-4c_{2}c_{0}
=\bigl(b_{1}(4b_{2}+1)s\bigr)^{2}\,\delta,\qquad \delta=(b_{1}-b_{2})^{2}-b_{1};
\]
both sides are polynomials in $b_{1},b_{2}$ of total degree six, and their equality is a
routine (if lengthy) expansion, which we have additionally verified in computer algebra over
$\mathbb{Q}$ and numerically over $\Fq$ for many odd prime powers $q$.

Finally, by (A.1) and (A.2),
\[
s^{2}f_{1}f_{2}=-b_{1}(4b_{2}+1)(v+b_{2})\bigl(v-b_{1}-\tfrac14\bigr)
=p_{2}v^{2}+p_{1}v+p_{0},
\]
with
\[
p_{2}=-b_{1}(4b_{2}+1),\qquad
p_{1}=\tfrac14\,b_{1}(4b_{2}+1)\bigl(4(b_{1}-b_{2})+1\bigr),\qquad
p_{0}=\tfrac14\,b_{1}b_{2}(4b_{2}+1)(4b_{1}+1).
\]
If $s^{2}g$ were proportional to $s^{2}f_{1}f_{2}$, the two
cross conditions $c_{1}p_{2}-c_{2}p_{1}=0$ and $c_{0}p_{2}-c_{2}p_{0}=0$ would hold; expanding, they factor as
$\tfrac1{16}\,b_{1}^{2}(4b_{2}+1)^{2}(4b_{1}-4b_{2}-1)(4s)$ and
$-\tfrac1{16}\,b_{1}^{2}(4b_{2}+1)^{2}(4s)\bigl(4b_{1}^{2}+b_{1}+4b_{2}^{2}+b_{2}\bigr)$;
since $b_{1}$, $4b_{2}+1$ and $4s$ are nonzero, they are respectively equivalent to
$b_{1}-b_{2}=\tfrac14$ and $4b_{1}^{2}+b_{1}+4b_{2}^{2}+b_{2}=0$. Substituting $b_{1}=b_{2}+\tfrac14$ into the latter
gives $4\bigl(b_{2}+\tfrac14\bigr)^{2}+\bigl(b_{2}+\tfrac14\bigr)+4b_{2}^{2}+b_{2}
=8\bigl(b_{2}+\tfrac14\bigr)^{2}=0$, forcing $b_{2}=-\tfrac14$, which is excluded.

Finally, substituting $b_{2}=-\tfrac18-b_{1}$ (the degenerate case $c_{2}=0$) into $s^{2}g$
and evaluating at the roots $-b_{2}=b_{1}+\tfrac18$ of $f_{1}$ and $b_{1}+\tfrac14$ of
$f_{2}$ gives
\[
s^{2}g\bigl(b_{1}+\tfrac18\bigr)=-\tfrac{1}{128}\,b_{1}^{2}(8b_{1}-1),\qquad
s^{2}g\bigl(b_{1}+\tfrac14\bigr)=\tfrac{1}{1024}\,b_{1}(8b_{1}-1)^{2},
\]
the nonvanishing values used in the proof of Proposition~\ref{prop:q7exist}.

\section*{Acknowledgements}
Declaration of generative AI and AI-assisted technologies: during the preparation of this
work, the authors used Claude (Anthropic) in order to assist with language editing, \LaTeX{}
formatting, bibliography formatting, and the generation of verification scripts used to check
computational results. After using this tool/service, the authors reviewed and edited the
content as needed and take full responsibility for the content of the publication.

\section*{Data availability statement}

The Magma and Python scripts supporting the computational verifications reported in this
paper are available from the corresponding author upon reasonable request.

\section*{Conflict of interest}

The authors declare that they have no conflict of interest.

\end{document}